\documentclass{llncs}

\usepackage{hyperref}
\usepackage{graphicx}
\usepackage{amsmath}
\usepackage{latexsym}
\usepackage{amsfonts}
\usepackage{xcolor}
\usepackage{url}
\usepackage{listings}
\usepackage{times}
\usepackage{wrapfig}
\usepackage{amsthm}
\usepackage{enumitem}

\newcommand{\OMIT}[1]{}
\newcommand{\ACT}{\ensuremath{\text{\sf ACT}}}

\newcommand{\defn}[1]{\textit{#1}} 
\newcommand{\mcl}[1]{\ensuremath{\mathcal{#1}}}
\newcommand{\set}[1]{\{#1 \}}
\newcommand{\sem}[1]{[\![#1 ]\!]}

\newcommand{\struct}{\mathfrak{S}}

\newcommand{\Tree}{\ensuremath{\text{\scshape Tree}}}

\newcommand{\transys}{\ensuremath{\langle S; \to\rangle}}
\newcommand{\transysMDP}{\ensuremath{\langle S = V_1 \cup V_2; {\to_1}, {\to_2} \rangle}}
\newcommand{\FMDP}{\ensuremath{\langle S = V_1 \cup V_2; {\to_1}, {\to_2}, \SF, \WF \rangle}}

\newcommand{\transysG}{\ensuremath{\langle S; \{\to_a\}_{a \in \ACT}\rangle}}
\newcommand{\Kripke}{\ensuremath{\langle S; {\{\to_a\}_{a \in \ACT}},\allowbreak \labeling\rangle}}
\newcommand{\transysDom}{S}

\newcommand{\labeling}[0]{\ell}
\newcommand{\AP}{\ensuremath{\text{\sf AP}}}

\newcommand{\DTMCAP}{\ensuremath{\langle S; \delta, \labeling\rangle}}
\newcommand{\transysAP}{\ensuremath{\langle S; \to,\allowbreak \labeling\rangle}}

\newcommand{\Rec}{\ensuremath{\text{\scshape Rec}}}

\newcommand{\empseq}{\ensuremath{\epsilon}}

\newcommand{\ialphabet}{\ensuremath{\Sigma}}

\newcommand{\Var}{\ensuremath{\mcl{V}}}

\newcommand{\N}{\ensuremath{\mathbb{N}}}

\newcommand{\Z}{\ensuremath{\mathbb{Z}}}

\newcommand\xor{\oplus}

\newtheorem*{convention}{Convention}

\newcommand{\problemx}[3]{
\par\noindent\underline{\sc#1}\par\nobreak\vskip.2\baselineskip
\begingroup\clubpenalty10000\widowpenalty10000
\setbox0\hbox{\bf Instance: }\setbox1\hbox{\bf Question: }
\dimen0=\wd0\ifnum\wd1>\dimen0\dimen0=\wd1\fi
\vskip-\parskip\noindent
\hbox to\dimen0{\box0\hfil}\hangindent\dimen0\hangafter1\ignorespaces#2\par
\vskip-\parskip\noindent
\hbox to\dimen0{\box1\hfil}\hangindent\dimen0\hangafter1\ignorespaces#3\par
\endgroup}

\newcommand{\Path}{\pi}

\newcommand{\Aut}{\ensuremath{\mathcal{A}}} 
\newcommand{\Tra}{\ensuremath{\mathcal{T}}} 
\newcommand{\AutB}{\ensuremath{\mathcal{B}}} 
\newcommand{\transrel}{\ensuremath{\Delta}} 
\newcommand{\tran}[1]{\ensuremath{\stackrel{#1}{\longrightarrow}}}
\newcommand{\controls}{\ensuremath{Q}} 
\newcommand{\finals}{\ensuremath{F}} 
\newcommand{\Lang}{\ensuremath{\mathcal{L}}} 
\newcommand{\AutRun}{\ensuremath{\rho}} 

\newcommand{\arity}{\ensuremath{\text{\texttt{ar}}}}

\newcommand{\rarw}{\rightarrow}

\newcommand{\ModelRun}{\ensuremath{\pi}}
\newcommand{\run}{\mathit{Run}}

\newcommand{\Size}[1]{\ensuremath{\|#1\|}}
\newcommand{\Length}[1]{\ensuremath{|#1|}}

\newcommand{\DEC}{\ensuremath{\text{\texttt{DEC}}}}
\newcommand{\ID}{\ensuremath{\text{\texttt{ID}}}}
\newcommand{\RES}{\ensuremath{\text{\texttt{RESET}}}}

\newcommand{\Prob}{\ensuremath{\text{Prob}}}

\newcommand{\funcClass}{\ensuremath{\mathfrak{F}}}
\newcommand{\SFtype}{\ensuremath{\text{\texttt{c}}}}
\newcommand{\WFtype}{\ensuremath{\text{\texttt{j}}}}

\newcommand{\CTR}{\mathcal{X}}
\newcommand{\proj}[1]{\ensuremath{\mathit{proj}_{#1}}}

\newcommand{\biword}[2]{(#1,#2)}

\newcommand{\Implies}{\Rightarrow}
\newcommand{\SF}{\mathfrak{C}} 
\newcommand{\WF}{\mathfrak{J}} 

\newcommand{\blinded}[1]{\textcolor{black!60}{[omitted for submission]}}
\newcommand{\toolname}[0]{\textsc{FairyTail}}
\newcommand{\slrp}[0]{\textsc{Slrp}}
\newcommand{\td}[1]{\textcolor{blue}{\ifmmode \text{[#1]}\else [#1] \fi}}

\newenvironment{to-do}
{ \rule{1ex}{1ex}\hspace{\stretch{1}} \bfseries}
{ \hspace{\stretch{1}}\rule{1ex}{1ex} \vspace{1ex}}

\newif\ifdraft\drafttrue
\ifdraft
\newcommand\anthony[1]{{\color{blue}
[#1 - \textbf{Anthony}]}}
\newcommand\rupak[1]{{\color{green}
[#1 - \textbf{RM}]}}
\newcommand\philipp[1]{{\color{orange}
[#1 - \textbf{PR}]}}

\else
\newcommand\todo[1]{}
\newcommand\anthony[1]{}
\newcommand\rupak[1]{}
\newcommand\philipp[1]{}

\fi

\title{Fair Termination for Parameterized Probabilistic Concurrent Systems
(Technical Report)}

\author{
Ond\v{r}ej Leng\'{a}l\inst{1} \and
Anthony W. Lin\inst{2} \and
Rupak Majumdar\inst{3} \and
Philipp R\"ummer\inst{4}}

\institute{
  {FIT, Brno University of Technology, Czech Republic}
\and
  {Department of Computer Science, University of Oxford, UK}
\and
  {MPI-SWS Kaiserslautern, Germany}
\and
  {Uppsala University, Sweden}
}

\authorrunning{O. Lengál, A. W. Lin, R. Majumdar, P. R\"ummer} 

\begin{document}

%


\maketitle

\vspace{-4mm}
\begin{abstract}
  We consider the problem of automatically verifying that a parameterized family 
of probabilistic concurrent systems terminates with probability one for 
all instances against adversarial schedulers.
A parameterized family defines an infinite-state system: for each number 
$n$, the family consists of an instance with $n$ finite-state processes.
In contrast to safety, the parameterized verification of
liveness is currently still considered extremely challenging especially in the 
presence of probabilities in the model.
One major challenge is to provide a sufficiently powerful symbolic framework.
One well-known symbolic framework for the parameterized verification of 
non-probabilistic concurrent systems is \emph{regular model checking}.
Although the framework was recently extended to probabilistic systems,
incorporating fairness in the framework --- often crucial for verifying 
termination --- has been especially difficult
due to the presence of an infinite number of fairness constraints (one for each process).
\OMIT{
In this paper, we investigate how to embed \emph{finitary fairness} 
(a
sufficiently powerful fairness condition
to this framework. 
}
\OMIT{
Since termination typically fails over concurrent systems without imposing
certain fairness conditions on the schedulers, we propose to tackle the problem
of embedding \emph{finitary fairness} in the regular model checking framework 
for parameterized probabilistic systems.
}
Our main contribution is a systematic, regularity-preserving, encoding of 
\emph{finitary fairness} (a realistic notion of fairness proposed by Alur \& 
Henzinger) in the framework of regular model checking for 
probabilistic
parameterized systems. Our encoding reduces termination with finitary fairness 
to verifying parameterized termination \emph{without fairness} 
over probabilistic systems in regular model checking (for which
a verification framework already exists).
\OMIT{
Combined with recent advances in termination checking, we 
obtain a fully automatic checker for parameterized probabilistic systems
with respect to process-fair termination properties.
}
We show that our algorithm 
could verify termination for many interesting examples
from distributed algorithms (Herman's protocol) and evolutionary biology 
(Moran process, cell cycle switch), which do not hold under the standard notion
of fairness.  To the best of our knowledge, our algorithm is the first
fully-automatic method that can prove termination for these examples.

\end{abstract}

\pagestyle{plain}
%

\vspace{-8.0mm}
\section{Introduction} \label{sec:intro}
\vspace{-1.0mm}

In parameterized probabilistic concurrent systems, 
a population of {\em agents}, each typically modeled as a finite-state probabilistic
program, run concurrently in discrete time and update their states based on 
probabilistic transition rules.
The interaction is governed by an underlying {\em topology}, which determines which agents
can interact in one step, and 
a~{\em scheduler}, which picks the specific agents involved in the interaction.
Concurrent probabilistic systems arise as models of distributed algorithms 
\cite{LSS94,Her90,IJ90,LR81,Fokkink-book},
where each agent is a processor, the interaction between processors is determined
by a communication topology, and the processor can update its internal state based
on the communication as well as randomization.
In each step, the scheduler adversarially chooses a~processor to run.
Concurrent probabilistic populations also arise in agent-based population models
in biology 
\cite{Lieberman05}, wherein an agent can represent an allele, a cell, or a~species, 
and the interaction between agents describes how these entities evolve over 
time.
\OMIT{
For such systems, typical properties of interests include \emph{termination},
i.e., whether the system is guaranteed to reach certain states with probability
1 (also called \emph{almost surely terminate}) regardless of the behaviour of 
the schedulers.
}
For a population of a~fixed size, there is a rich theory of
probabilistic verification \cite{CY92,BK08,PRISM,Vardi85} based on
finite-state Markov decision processes (MDPs).  Verification
questions for population models, however, ask if a property holds for
populations of \emph{all} sizes: even if each agent is finite-state, the
family of all processes (for each population size) is an
infinite-state MDP.  
Indeed, for many simple population models, one
can show that the verification question is undecidable, even for
reachability or safety properties in the non-probabilistic 
setting~\cite{AK86,BertrandF13,Javier16}.
Consequently, the verification question for populations requires
techniques beyond finite-state probabilistic verification, and
requires symbolic techniques to represent potentially infinite sets of
states.

One well-known symbolic framework for verifying parameterized
non-probabilistic
concurrent systems is \emph{regular model checking} 
\cite{Parosh12,rmc-thesis,LTL-MSO,BLW03,Neider13,TL10}, where states of a population are modeled using 
words over
a suitable alphabet, sets of states are represented as regular
languages, and the transition relation is defined as a regular
transducer.  From parameterized verification of non-probabilistic
processes, it is known that regular languages provide a robust symbolic
representation of infinite sets, and automata-theoretic algorithms
provide the basis of checking safety or termination properties.

In this paper, we consider the problem of verifying that a given parameterized 
family of probabilistic concurrent systems 
\emph{almost surely terminates}, i.e., reaches certain final states with 
probability 1 from each initial state regardless of the behaviour of the 
schedulers. 
Termination is a fundamental property when verifying parameterized 
probabilistic systems.
Since termination typically, however, fails without imposing certain \emph{fairness} conditions on the 
scheduler, it is crucial to be able to incorporate fairness assumptions into
a termination analysis. 
Therefore, although the framework of regular
model checking has recently been extended for proving termination (without
fairness) over
parameterized probabilistic concurrent systems \cite{LR16}, it still cannot be 
used to prove termination for many interesting parameterized probabilistic
concurrent systems. 
%

\emph{What notion of fairness should we consider for proving termination 
for parameterized probabilistic concurrent systems?} 
To answer this question, one would naturally start by looking at standard
notions of fairness in
probabilistic model checking \cite{BK08}, which asserts that every process must
be chosen infinitely often. 
However, this notion seems to be too weak to prove
termination for many of our examples, notably Herman's self-stabilizing
protocol \cite{Her90} in an asynchronous setting, 
and population models from biology (e.g. Moran's process \cite{Lieberman05}). 
The standard notion of fairness gives rise to a rather unintuitive and
unrealistic strategy for the scheduler, which could 
delay an enabled process for as long as it desires while still being fair (see 
\cite[Example 8]{BKL14}
and the Herman's protocol example in Section \ref{sec:semantics}).
For this reason, we propose to consider Alur \& Henzinger's \cite{AH98}
\defn{finitary fairness} --- a stronger notion of fairness that allows 
the scheduler to delaying executing an enabled process in an infinite run for
at most $k$ steps, for some unknown but fixed bound $k \in \N$. 
Alur \& Henzinger argued that
this fairness notion is more realistic in practice, but it is not as restrictive
as 
the notion of \defn{$k$-fairness}, which fixes the bound $k$ a priori. 
In addition, it should be noted that finitary fairness is strictly weaker than 
probabilistic fairness (scheduler chooses processes randomly) 
for almost-sure termination over finite MDPs 
and parameterized probabilistic systems (an infinite family of finite MDPs).
We will show in this paper that there are many interesting examples of
parameterized probabilistic concurrent systems for which termination is
satisfied under finitary fairness, but \emph{not} under the most general notion
of fairness.
\smallskip
\noindent
\textbf{Contributions.}
Our main contribution is a systematic, regularity-preserving,
encoding of finitary fairness in the framework of regular model checking for
parameterized probabilistic concurrent systems. More precisely,
our encoding reduces the problem of verifying almost sure termination under 
finitary fairness to almost sure termination \emph{without fairness} in regular 
model checking, for which a verification framework exists \cite{LR16}.

In general, the difficulty with finding an encoding of fairness is how to deal 
with an infinite number of fairness requirements (one for each process) in a 
systematic and regularity-preserving manner. There are known encodings of
general notions of fairness in regular model checking, e.g., by using a token
that is passed to the next process (with respect to some ordering of the
processes) when the current process is executed,
and ensuring that the first process holds the token and passes it to the right
infinitely many times
(e.g. see \cite{LTL-MSO,rmc-thesis}). However, these encodings do not work in 
in our case for
several reasons. Firstly, they do not take into account the unknown upper bound 
(from finitary fairness) within which time a process has to be executed. 
Adapting
these encodings to finitary fairness would require \defn{the use of unbounded
counters}, which do not preserve regularity.  Secondly,
such encodings would yield the problem of verifying an almost-sure Rabin 
property (of the form $\Box\Diamond A \wedge \Diamond B$ in LTL notation,
where $A$ and $B$ are regular sets). Although we could reduce this to an
almost-sure termination property by means of product automata construction (i.e.
by first converting the formula to deterministic Rabin automaton), the target 
set $B$ in the resulting termination property $\Diamond B$ (consisting of 
configurations in strongly connected components satisfying some properties)
is \emph{not} necessarily regular. 

Instead, we revisit the well-known \emph{abstract program
transformation} in the setting of non-probabilistic concurrent systems 
\cite{Francez} encoding fairness into the program by associating to each 
process an unbounded counter that acts
as an ``alarm clock'', which will ``set off'' if an enabled process has not been
chosen by the scheduler for ``too long.'' This abstract program
transformation has been adapted by Alur \& Henzinger \cite{AH98} in the case of
finitary fairness by additionally incorporating an extra counter $n$ that 
stores the unknown upper bound and resetting the value of a counter belonging to
a chosen process to the ``default value'' $n$. Our contributions are as follows:
\begin{enumerate}
\item We show how Alur \& Henzinger's program transformation could be adapted to
the setting of probabilistic parameterized concurrent systems (infinite family 
of finite MDPs). This involves constructing a new parameterization of the system
(using the idea of weakly finite systems) and a proof that the transformation
preserves reachability probabilities.
\item We show how the resulting abstract program transformation could be made
    concrete in the setting of regular model checking \emph{without using 
    automata models beyond finite automata}.
\item We have implemented this transformation in \toolname{}. Combined with the
    existing algorithm \cite{LR16} for verifying almost sure termination 
        (without fairness) in regular model checking, we have successfully
        verified a number of models obtained from distributed algorithms and
        biological systems including Herman's protocol~\cite{Her90}, Moran
        processes in a linear array~\cite{Moran58,Lieberman05}, and 
        the cell cycle switch
        model~\cite{cell-cycle-switch} on ring and line topologies.
        To the best of our knowledge, our algorithm is the first 
        fully-automatic method that can prove termination for these examples. 
\end{enumerate}

\smallskip
\noindent
\textbf{Related work.}
There are few techniques for automatic 
verification of liveness properties of parameterized probabilistic
programs.
Almost sure verification of probabilistic finite-state programs goes
back to Pnueli and co-workers \cite{HSP83,PnueliZ86}.
Esparza et al.~\cite{EGK12} generalize the reasoning to weakly finite programs,
and describe a~heuristic to guess a \emph{terminating pattern} 
by constructing a nondeterministic program from
a given probabilistic program and a terminating pattern candidate. 
This allows them to exploit model checkers and termination provers for nondeterministic programs. 
More recently, Lin and R\"ummer \cite{LR16} consider unconditional termination 
for parameterized probabilistic programs. 
While our work builds on these techniques,
our main contribution is the incorporation of fairness in regular model checking
of probabilistic programs, which was not considered before.

Fairness for concurrent probabilistic systems was considered by Vardi \cite{Vardi85}
and by Hart, Sharir, and Pnueli \cite{HSP83}, and generalized later 
\cite{PnueliZ86,deAlfaro97,BaierKwiatkowska98}.
The focus was, however, on a fixed number of processes. 
The notion of fairness through explicit scheduling was developed by Olderog and Apt
\cite{OlderogA88}. More recently, notions of fairness for infinitary control
(i.e., where an infinite number of processes can be created) was considered by
Hoenicke, Olderog, and Podelski \cite{OlderogP10,HoenickeOP10}.

\OMIT{

We reduce fair termination to checking realizability of a fair scheduler and to checking
unconditional termination.
For the former, there is a recent algorithm for safety games over rational graphs (an extension
of automatic graphs) proposed by Neider and Topcu \cite{NT16}. 
For the latter, there is a recent incremental regular model checking
algorithm for reachability games proposed by Lin and R\"ummer \cite{LR16}.
}

Martingale techniques have been used to prove termination of sequential,
infinite-state, probabilistic programs \cite{Sriram13,Mon01,FioritiH15,Chakarov,KaminskiKMO16}.
These results are not comparable to our results, as they do not consider
unbounded families of fairness constraints nor communication topologies.

%

\vspace{-2.0mm}
\section{Preliminaries} \label{sec:prelim}
\vspace{-1.0mm}

\noindent
\textbf{General notations}:
For any two given real numbers $i \leq j$, we use a standard notation
(with an extra subscript) to denote real intervals, e.g., $[i,j]_{\mathbb{R}} = 
\{ k \in \mathbb{R} : i \leq k \leq j \}$ and $(i,j]_{\mathbb{R}} = \{ k \in \mathbb{R} : i < 
k \leq j \}$. We will denote intervals over integers by removing the
subscript, i.e., $[i,j] = [i,j]_{\mathbb{R}} \cap \Z$.
Given a set $S$, we use 
$S^*$ to denote the set of all finite
sequences of elements from $S$. The set $S^*$ always includes the empty
sequence, which we denote by $\empseq$. 
We use $S^+$ to denote the set $S^* \setminus \{\empseq\}$.
Given two sets of words
$S_1, S_2$, we use $S_1\cdot S_2$ to denote the set $\{ v\cdot w: v\in S_1,
w\in S_2\}$ of words formed by concatenating words from $S_1$ with words from
$S_2$. Given two relations $R_1,R_2 \subseteq S \times S$, we define
their composition as $R_1 \circ R_2 = \{ (s_1,s_3) :  \exists s_2 ((s_1,s_2)
\in R_1 \wedge (s_2,s_3) \in R_2)\}$. 
\smallskip

\noindent
\textbf{Transition systems}:
We fix the (countably infinite) set $\AP$ of \defn{atomic propositions}. 
Let $\ACT$ be a~finite set of \defn{action symbols}.
A~\defn{transition system} over $\ACT$ is
a tuple $\struct = \Kripke$,
where $S$ is a set of \defn{configurations}, $\to_a\ \subseteq S \times S$ 
is a binary relation over $S$,
and $\labeling : \AP \to 2^{S}$ maps atomic propositions to sets of configurations
(we omit $\labeling$ if it is not important).
We use $\to$ to denote the relation $\left(\bigcup_{a \in \ACT} \to_a\right)$. 
The notation $\to^+$ (resp. $\to^*$) is used to denote the transitive (resp.
transitive-reflexive) closure of $\to$. 
We say that a sequence $s_1 \to \cdots \to s_n$ is a \defn{path} (or
\defn{run}) in $\struct$ (or in $\to$). Given two paths $\ModelRun_1: s_1
\to^* s_2$ and $\ModelRun_2: s_2 \to^* s_3$ in $\to$, we may concatenate them
to obtain $\ModelRun_1 \odot \ModelRun_2$ (by gluing together~$s_2$).
We call $\ModelRun_1$ a \defn{prefix} of $\ModelRun_1 \odot \ModelRun_2$.
For each $S' \subseteq S$, we use the notations $\mathit{pre}_{\to}(S')$ and 
$\mathit{post}_{\to}(S')$ to denote the pre/post image of $S'$ under $\to$.
That is, $\mathit{pre}_{\to}(S') = \{ p \in S : \exists q\in S'( p \to q ) \}$ and
$\mathit{post}_{\to}(S') = \{ q \in S : \exists p\in S'( p \to q ) \}$.

\OMIT{
When dealing with probabilistic systems, we will find the following notations
handy. For two sets $S, S' \subseteq S$ of configurations in $\struct$, denote
by $\Paths_\struct(S,S')$ the set of all paths from (some state in) $S$ to 
(some state in) $S'$. We will omit mention of $\struct$ if $\struct$ is
clear from the context.
\anthony{Need to check if we can remove this notation}
}
\OMIT{
Given a relation $\to \subseteq S \times S$ and subsets 
$S_1,\ldots,S_n \subseteq S$, denote by 
$\Rec_{\to}(\{S_i\}_{i=1}^n)$ to be the set of elements $s_0 \in S$ for which
there exists an infinite path $s_0 \to s_1 \to \cdots$ visiting
each $S_i$ infinitely often, i.e., such that, for each
$i\in[1,n]$, there are infinitely many $j \in \N$ with $s_j \in S_i$.
}
\smallskip


%

\noindent
\textbf{Words and automata}:
We assume basic familiarity with finite word automata.
Fix a finite alphabet $\Sigma$. For each finite word $w = w_1\ldots w_n \in 
\Sigma^*$, we
write $w[i,j]$, where $1 \leq i \leq j \leq n$, to denote the segment
$w_i\ldots w_j$. Given an automaton $\mcl{A} = (\Sigma,Q,\delta,q_0,F)$,
a run of $\mcl{A}$ on $w$ is a function $\rho: \{0,\ldots,n\}
\rarw Q$ with $\rho(0) = q_0$ that obeys the transition relation $\delta$. 
We may also denote the run $\rho$ by the word $\rho(0)\cdots \rho(n)$ over
the alphabet $Q$. 
The run $\rho$ is said to be \defn{accepting} if $\rho(n) \in F$, in which
case we say that~$w$ is \defn{accepted} by $\mcl{A}$. The language
$L(\mcl{A})$ of $\mcl{A}$ is the set of words in $\Sigma^*$ accepted by
$\mcl{A}$.

\smallskip
\noindent
\textbf{Reachability games}:
We recall some basic concepts on 2-player reachability games 
(see e.g. \cite[Chapter 2]{ALG-book} on games with 1-accepting 
conditions).
An \defn{arena} is a~transition system $\struct = \transysMDP$, where $S$
(i.e. the set of ``game configurations'')
is partitioned into two disjoint sets $V_1$ and $V_2$ such that $pre_{\to_i}(S)
\subseteq V_i$ for each $i \in \{1,2\}$. The transition relation $\to_i$ 
denotes the actions of Player~$i$.
Similarly, for each $i \in \{1,2\}$, the configurations $V_i$ are controlled by 
Player~$i$. In the following, Player 1 will
also be called ``Scheduler,'' and Player~2 ``Process''. Given a set
$I_0 \subseteq S$ of initial configurations and a~set $F \subseteq S$ of 
final (a.k.a. target) configurations, the goal of Player~2 is to reach $F$
from $I_0$,
while the goal of Player~1 is to avoid it. 
More formally, a \defn{strategy} for Player $i$ is
a~partial function $f: S^*V_i \to S$ such that, for each $v \in S^*$ and
$p \in V_i$, if $vp$ is a~path in $\struct$ and $p$ is not a dead end 
(i.e., $p \to_i q$ for some $q$), then
$f(vp)$ is defined in such a way that $p \to_i f(vp)$. Given a strategy $f_i$
for Player~$i \in \{1,2\}$ and an initial configuration $s_0 \in S$, we can define a~unique
(finite or infinite) path in $\struct$ 
such that $\pi: s_0 \to_{j_1} s_1 \to_{j_2} \cdots$
where $s_{j_{k+1}} = f_i(s_0s_1\ldots s_{j_k})$ for $i \in \{1,2\}$
is the (unique) configuration s.t.~$s_{j_k} \in V_i$. 
Player~2 \emph{wins} iff some configuration 
in $F$ appears in $\pi$, or if the path is finite and the last configuration
belongs to Player~1. Player~1 \emph{wins} iff Player~2 does not
win; we say Player~2 \emph{loses}. A~strategy $f$ for Player $i$ 
is \defn{winning} from $I_0$ if for each strategy 
$g$ of Player~$3-i$, the unique path in $\struct$ from each 
$s_0 \in I_0$
witnesses a win for Player~$i$. Such games (a.k.a.~\emph{reachability games})
are \emph{determined} (see e.g. \cite[Proposition 2.21]{ALG-book}): 
either Player~1 has a winning strategy or Player~2 has a winning strategy.

%


%
\OMIT{
later, to deal with
liveness over parameterised systems it suffices to deal with the following
computational problem for reachability games: given an arena $\struct = 
\transysG$, a set $I_0 \subseteq S$ of initial configurations, 
and a set $F \subseteq S$ of final configurations, decide whether
Player 2 has a winning strategy from \emph{all of} $I_0$ to reach $F$ in 
$\struct$? [In other words, whether $I_0$ is a subset of the winning 
configurations for Player 2.] 
}

\begin{convention}
    For notational simplicity, w.l.o.g., we make the following assumptions on our reachability
    games. They suffice for the purpose of proving liveness for parameterised
    systems. 
    \begin{description}
        \item[(A0)] Arenas are \defn{strictly alternating}, i.e., 
        a move made by a player does not take the game back to her 
        configuration ($\mathit{post}_{\to_i}(S) \cap V_i = \emptyset$, for each 
        $i \in \{1,2\}$).
    \item[(A1)] 
        Initial and final configurations belong to
        Player 1, i.e., $I_0, F \subseteq V_1$
    \item[(A2)] Non-final configurations are not dead ends:
        $\forall x\in S\setminus F, \exists y: x \to_1 y \vee x \to_2 y$.
    \end{description}
    \label{con:games}
\end{convention}

\smallskip
\noindent
\textbf{Markov chains:}
A (discrete-time) \defn{Markov chain} (a.k.a. \defn{DTMC}) 
is a structure of the form $\struct = \DTMCAP$ where $\transysDom$ is
a set of configurations, $\delta$ is a function that associates a 
configuration $s \in \transysDom$ with a probability distribution over a
sample space $D \subseteq \transysDom$ (i.e. the probability of going to a 
certain configuration from $s$), and $\labeling:\AP \to 2^{S}$ maps
atomic propositions
to subsets of $\transysDom$.
In what follows, we will assume 
that each
$\delta(s)$ is a~discrete probability distribution with a finite sample space.
This assumption allows us to simplify our notation: a~DTMC~$\DTMCAP$ can be
seen as a~transition system $\transysAP$ with a transition probability function
$\delta$ mapping a transition $t = (s,s') \in\ \to$ to a value $\delta(t) \in 
(0,1]$ such that $\sum_{s' \in \mathit{post}(s)} \delta((s,s')) = 1$.
That is, transitions with zero probabilities are removed from $\to$.
We will write $s \tran{p} s'$ to denote $s \to s'$ and that $\delta((s,s')) = p$. 
The \defn{underlying transition graph} of a~DTMC~$\DTMCAP$ is the transition system
$\transysAP$ with $\delta$ omitted.
Given a finite path $\ModelRun = s_0 \to \cdots \to s_n$ from the initial configuration
$s_0 \in S$, let $\run_{\ModelRun}$ be the set of all finite/infinite paths with
$\ModelRun$ as a prefix, i.e., of the form $\ModelRun \odot \ModelRun'$ for 
some finite/infinite
path $\ModelRun'$. 
Given a set $F \subseteq S$ of
target configurations, the probability $\Prob_{\struct}(s_0 \models \Diamond F)$ 
(the subscript $\struct$ may be omitted when understood)
of reaching $F$ 
from $s_0$ in $\struct$ can be defined using a standard cylinder construction 
(see e.g \cite{marta-survey}) as follows. 
For each finite path $\pi = s_0 \to \cdots \to s_n$ 
in $\struct$ from $s_0$, we set $\run_{\ModelRun}$ to be a basic cylinder,
to which we associate the probability $\Prob(\run_{\ModelRun}) =
\prod_{i=0}^{n-1} \delta((s_i,s_{i+1}))$.
This gives rise to a unique probability measure for the $\sigma$-algebra
over the set of all runs from $s_0$. The probability 
$\Prob(s_0 \models \Diamond F)$ is then the probability of the event~$F$
containing all paths in $\struct$ with some ``accepting'' finite path as
a prefix, i.e., a finite path from $s_0$ ending in some configuration in $F$.
In general, given an LTL formula $\varphi$ over $\AP$, the event containing all 
paths from $s_0$ in $\struct$ satisfying $\varphi$ is measurable 
\cite{Vardi85} and its probability value $\Prob(s_0 \models \varphi)$ 
is well-defined.

\noindent
\textit{Notation:}
Whenever understood, we will omit 
mention of $\labeling$ from $\DTMCAP$.

\vspace{-2.0mm}
\section{Abstract Models of Probabilistic Concurrent Programs}\label{sec:semantics}
\vspace{-1.0mm}

In this section, we recall the notion of Markov Decision Processes (MDPs) and
fair MDPs \cite{BK08}. These serve as our abstract models of probabilistic
concurrent programs. We then define the notion of finitary fairness 
\cite{AH98} and discuss its basic properties in the setting of MDPs.

\subsection{Markov Decision Processes}
A \defn{Markov decision process} (\defn{MDP}) is a 
strictly alternating arena $\struct = \transysMDP$ such that $\langle S; 
\to_2 \rangle$ is a DTMC, i.e., $\to_2$ is associated with some transition
probability function, and that the atomic propositions are not important. 
Intuitively, the transition relation $\to_1$ is nondeterministic
(controlled by a ``demonic'' scheduler), whereas the transition relation 
$\to_2$ is 
probabilistic. By definition of arenas, the configurations of the MDPs are
partitioned into the set $V_1$ of \emph{nondeterministic states} (controlled by 
Scheduler) and the set $V_2$ of \emph{probabilistic states}. Formally,
$pre_{\to_1}(S) \cap pre_{\to_2}(S) =~\emptyset$.
Each Scheduler's strategy\footnote{Also called ``scheduler'' or 
``adversary'' for short.} $f: S^* V_1 \rightarrow S$ gives rise to an 
infinite-state DTMC with the underlying transition system
$\struct_f = \langle S'; \to_3, \labeling \rangle$ and the transition probability function
$\delta'$ defined as follows. Here, $S'$ is the set of all 
finite/infinite paths $\pi$ from $s_0$. For each state $s' \in S$ and each path 
$\pi$ from $s_0$ ending in some state $s \in S$,
%
we define $\pi \to_3 \pi s'$ iff:
(1) if $s \in V_1$ is a nondeterministic state, then $f(\pi) = s'$, and
(2) if $s \in V_2$ is a probabilistic state, then $s \to_2 s'$. 
Intuitively, $\struct_f$ is an unfolding of the game arena $\struct$ (i.e. a 
disjoint union of trees) where branching only occurs on probabilistic states.
Transitions
$\pi \to_3 \pi s'$
satisfying 
Case (1) have the probability $\delta'((\pi,\pi s')) = 1$; otherwise, its 
probability is $\delta'((\pi,\pi s')) = \delta((s,s'))$. We let $\labeling$ be
a function mapping each subset $X \subseteq \transysDom$ (used as an atomic
proposition) to the set of all finite paths in $\struct_f$ from $s_0$ to
$X$.
Since $\struct_f$
is a DTMC, given an LTL formula $\varphi$ over subsets of $\transysDom$ as 
atomic propositions, the probability $\Prob_{\struct_f}(s_0 \models \varphi)$ 
of satisfying $\varphi$ in $\struct$ from $s_0$ under the scheduler $f$ is 
well-defined. In particular, $\Prob_{\struct_f}(s_0 \models \Diamond F)$
is the probability of reaching $F$ from $s_0$ in $\struct$ under the scheduler 
$f$. The probability $\Prob_{\struct,\mathcal{C}}(s_0 
\models \varphi)$ 
of satisfying $\varphi$ from $s_0$ in the MDP $\struct$ under a 
class $\mathcal{C}$ of schedulers is defined
to be the infimum of the set of all probabilities $\Prob_{\struct_f}(s_0 
\models \varphi)$ over all $f \in \mathcal{C}$. We will omit mention of
$\mathcal{C}$ when it denotes the class of all schedulers.

An MDP is \defn{weakly-finite} \cite{EGK12} if from each configuration,
the set of all configurations that are reachable from it (in the
underlying transition system of the MDP) is finite. 
Note that the state space of weakly-finite MDPs can be infinite.  
The restriction of weak finiteness is another way of defining the notion of
\defn{parameterized
systems}, which are an infinite family of finite-state systems. 
Weakly-finite MDPs capture many interesting probabilistic concurrent systems 
in which each process is finite-state; this is the case for 
many probabilistic distributed protocols.

\OMIT{
In this paper, we are only concerned with the following \defn{liveness problem
for MDP}: given an MDP $\struct = \transysMDP$, a set $I_0 \subseteq S$ of 
initial states, and a set $F \subseteq S$ of target states, determine whether 
$\Prob_{\struct}(s_0 \models \Diamond F) = 1$ for every initial state
$s_0 \in I_0$. Note that this is equivalent to proving that
$\Prob_{\struct_f}(s_0 \models \Diamond S^*F) = 1$ for every initial state
$s_0 \in I_0$, and every scheduler $f$. 
Such a problem (which has many other
names: probabilistic universality, almost-sure probabilistic reachability, and 
almost-sure liveness) is commonly studied
in the context of MDPs (e.g. see \cite{CY95,Vardi85,LS07,marta-survey}).
In the case of probabilistic
parameterised systems
 $\mathcal{F} = \{\struct_i\}_{i \in \N}$, where each $\struct_i$ is a
 finite-state MDP, we may view $\mathcal{F}$ as an infinite-state MDP defined by
the disjoint union of $\struct_i$ over all $i \in \N$.  In this way, proving
liveness for $\mathcal{F}$ simply means proving liveness for \emph{each}
instance $\struct_i$ in $\mathcal{F}$.
}

\subsection{Fair Markov Decision Processes}

A \defn{fair Markov decision process (FMDP)} is a structure of the form
$\struct = \FMDP$, where $\transysMDP$ is an MDP, $\WF$ is a weak fairness 
(a.k.a. \emph{justice}) requirement, and $\SF$ is a strong 
fairness (a.k.a. \emph{compassion}) requirement. More precisely, a \defn{weak 
fairness requirement} is a  set (at most countably infinite) of \defn{atomic 
weak 
fairness requirements} of the form $\Diamond \Box A \Implies \Box\Diamond B$,
for some $A, B \subseteq \transysDom$. 
Here, the $\Box$ and $\Diamond$ 
modalities are the standard ``always'' 
and ``eventually'' LTL operators.
The set $A$ (resp.~$B$) will be called the \defn{premise} 
(resp.~\defn{consequence}).
\OMIT{
Intuitively, each atomic weak
fairness requirement restricts schedulers to paths on which configurations in $B$ must
appear infinitely often, whenever it is the case that after some point only 
configurations in $A$ are seen on the paths. 
}
Intuitively, if $A$ is interpreted as ``Process 1
is waiting to move'' and $B$ as ``Process 1 is chosen'', then this 
fairness requirement
may be read as: at no point can Process 1 be continuously waiting to move 
without being chosen.
In addition, a \defn{strong
fairness requirement} is a set (again, at most countably infinite) of \defn{atomic 
strong 
fairness requirements} of the form $\Box \Diamond A \Implies \Box\Diamond B$,
for some $A, B \subseteq \transysDom$. 
\OMIT{
Loosely speaking, each atomic strong
fairness requirement restricts schedulers to paths on which configurations in 
$B$ must appear infinitely often, whenever it is the case that $A$ appears
infinitely often. 
}
Using the above example,
a strong fairness requirement reads: if Process 1 is waiting to move infinitely 
often, then it is chosen infinitely often. 
\OMIT{Observe that strong fairness implies weak fairness.}
As before, the set $A$ (resp.~$B$) will be called the \defn{premise} 
(resp.~\defn{consequence}).
In the following, when it is clear whether a fairness requirement is a~justice or 
a~compassion, we will denote it by the pair $(A,B)$ of premise and consequence.

Given an FMDP $\struct = \FMDP$, a configuration $s_0 \in \transysDom$, and
a scheduler $f$, since each atomic fairness requirement is an LTL formula and
there are at most countably many atomic fairness requirements, the set of paths 
from $s_0$ in the DTMC $\struct_f$ induced by $f$ satisfying $\SF$ and $\WF$ is 
measurable. We say that a scheduler $f$ is \defn{$\struct$-fair} if 
$\Prob_{\struct_f}(s_0 \models \SF \wedge \WF) = 1$ for every initial
configuration $s_0$. The fairness conditions $(\SF,\WF)$ are 
\defn{realizable} in $\struct$ if there exists at least one $\struct$-fair
scheduler.

A natural fairness notion we consider in this paper is \emph{process fairness}, which
asserts that each process is chosen infinitely often. 
For this notion of fairness, we can assume that the
consequence $B$ of each atomic fairness requirement asserts that a particular 
process is chosen. 
We make one simplifying assumption: 
\emph{each process is always enabled} (i.e., can always be chosen by the scheduler). 
This assumption is reasonable since we can always introduce an idle transition 
for each process.
Under this assumption, we have that
\emph{from each $v_1 \in V_1$, there 
exists a transition $v_1 \to_1 v_2$ for some $v_2 \in B$}. 
\OMIT{
    asserting that every process is chosen infinitely often.} 
This implies that our fairness conditions are always realizable and that
the probability $\Prob_{\struct,\mathcal{C}}(E)$ of event $E$ over 
the set of all $\struct$-fair schedulers is well-defined. 
\OMIT{
In
the sequel, for an FMDP $\struct$, we will simply denote 
$\Prob_{\struct,\mathcal{C}}(E)$ as $\Prob_{\struct}(E)$.
}

\OMIT{
Given a measurable set $E$ of paths from $s_0$ and an FMDP $\struct = \FMDP$, 
the probability $\Prob_{\struct}(E)$ can be defined as the probability 
$\Prob_{\struct,\mathcal{C}}(E)$, where $\mathcal{C}$ is the class of 
all schedulers that satisfy $\SF$ and
$\WF$ (i.e. the induced DTMCs contain only paths that satisfy $\WF$ and
$\SF$). 
}

\OMIT{
We consider the important special case of \defn{process fairness}, which
asserts that every process is chosen infinitely often. 
This implies that 
\emph{Scheduler may choose a process at will}, i.e., the consequence set of
each fairness constraint in the fair MDP $\struct$ belongs to $V_2$ (Process)
and can be satisfied within one step from each node in $V_1$. 
\anthony{This requirement contradicts our definition of fairness, where each
premise and consequence sets are subsets of $V_1$. What is the best way to
express this? Shall we say just write a formula like $\forall x\exists y\forall
x'( x \in A \Longrightarrow x \longrightarrow y \wedge y \longrightarrow x'
\wedge x' \in B)$?}
We call such
consequence sets \defn{chosen sets}. 
}

\subsection{Finitary Fairness}
Given an FMDP $\struct = \FMDP$, a configuration $s_0 \in \transysDom$, 
and a number $k \in \N$,
we say that a scheduler $f$ is \defn{$\struct$-$k$-fair} (or \defn{$k$-fair}
whenever $\struct$ is understood) if for each atomic fairness requirement 
$(A,B)$:
\begin{enumerate}
    \item if $(A,B)$ is justice, then (the underlying graph of) $\struct_f$ 
        contains no path $\pi$ of length $k$ satisfying the LTL formula
        $\Box (A \wedge \neg B)$.
    \item if $(A,B)$ is compassion, then $\struct_f$ contains no path
        $\pi$ satisfying the LTL formula $\psi_k \wedge \Box \neg B$, where
        $\psi_0 := \mathit{true}$ and $\psi_i := \Diamond (A \wedge \psi_{i-1})$ for
        each $i > 0$.
\end{enumerate}
In other words, a premise in a justice requirement cannot be satisfied for $k$
\emph{consecutive} steps without satisfying a consequence, while a premise in a
compassion requirement cannot be satisfied for $k$ (not necessarily consecutive)
steps without satisfying a consequence. A scheduler is said to be \defn{finitary
fair (fin-fair)} if it is $k$-fair for some $k$. 
The fairness conditions $(\SF,\WF)$ are said to be 
\defn{finitary-realizable (fin-realizable)} in $\struct$ if there exists at 
least one fin-fair scheduler. 
Under this assumption, the probability $\Prob_{\struct,\mathcal{C}}(E)$ of 
an event $E$ over the set $\mathcal{C}$ of all fin-fair schedulers is 
well-defined. In what follows, for an FMDP $\struct$, we will simply denote 
$\Prob_{\struct,\mathcal{C}}(E)$ as $\Prob_{\struct}(E)$.
In this paper, we propose to
study \defn{termination of probabilistic concurrent programs under finitary 
fairness}, i.e., to~determine whether $\Prob_{\struct,\mathcal{C}}(s_0 \models \Diamond
\finals) = 1$, where $\mathcal{C}$ is the class of all fin-fair schedulers.

The following proposition states one special property of weakly-finite MDPs.
\begin{proposition}
Let $\struct$ and $\struct'$ be two weakly-finite fair MDPs with identical 
    underlying transition systems (but possibly different probability values). 
    For each set $\finals$ of final states, 
    and each initial
    configuration $s_0$, it is the case that $\Prob_{\struct}(s_0 \models 
    \Diamond \finals) = 1$ iff $\Prob_{\struct'}(s_0 \models \Diamond \finals) = 1$.
    \label{prop:removeProb}
\end{proposition}

\begin{proof}
This proposition can be proved using basic machineries from 
probabilistic model checking \cite{BK08}.
Consider the finite MDPs $\struct^1$ and $\struct^2$ that are obtained from
$\struct$ and $\struct'$ by removing configurations in $\finals$ and their
fairness requirements $\varphi$.
It suffices to prove the following: for all schedulers
$\Prob_{\struct^1_\sigma}(s_0 \models \varphi) = 0$ iff for all schedulers
$\sigma$ $\Prob_{\struct^2_\sigma}(s_0 \models \varphi) = 0$.
This follows from standard results from probabilistic model checking
\cite[Theorem 10.122]{BK08} since $\varphi$ is a limit linear-time property.
\end{proof}

By Proposition~\ref{prop:removeProb}, when dealing with 
almost-sure finitary-fair termination of weakly-finite MDPs, we only care whether a transition 
has a zero or a non-zero probability, i.e., if it is non-zero, then the exact
value is irrelevant. Incidentally, the same also holds for other properties 
including
almost-sure termination without fairness and qualitative temporal specifications
\cite{HSP83,PnueliZ86,LR16}.
\emph{For this reason, we may simply omit these probability values from our 
symbolic representation of weakly-finite MDPs, which we will do from the next 
section onwards.} 
\OMIT{
A similar observation was already made for almost-sure
termination for probabilistic parameterised concurrent systems in 
\cite{LR16}.
}

\vspace{-0.0mm}
\subsection{Herman's Protocol} \label{sec:herman}
\vspace{-0.0mm}

Herman's protocol~\cite{Her90} is a distributed self-stabilization algorithm 
for a population of processes organized in a ring. The \defn{correct} 
configurations are those where exactly one process holds a token.
If, through some error, the ring enters an \defn{erroneous} configuration (in 
which multiple processes hold tokens), Herman's protocol ensures that the 
system will \defn{self-stabilize}: it will
almost surely go back to a configuration with only one token.

Let us discuss how the protocol works in more detail.
Fix $N\geq 3$ processors organized in a ring.
\OMIT{
Each process may or may not hold a token.
In each time step, an adversarial scheduler picks a process.}
If a~chosen process does not hold a token, then it can perform an idle
transition (i.e. do nothing). If a chosen process holds a token, then it can 
keep holding the token with probability $\frac{1}{2}$
or pass it on to its clockwise neighbor (the process $(i+1)\bmod N$, for processes numbered
$0,\ldots,N-1$) with probability~$\frac{1}{2}$. 
If a~process currently holds a token and receives another token from its (counter-clockwise) neighbor,
then the two tokens are merged\footnote{Herman \cite{Her90} describes a more 
general protocol in which tokens can be merged/destroyed with some probability. 
We consider this restriction for simplicity of presentation.}
into one, leaving the process with one token.

Formally, Hermann's protocol can be modeled as a weakly-finite 
Markov decision process whose states are vectors in $\set{\bot,\top}^*$.
For each $N$, the state of the protocol is described by a vector of $N$ bits,
with the $i$-th bit being 1 iff the $i$-th process holds a token.
From a state $\mathbf{v}$, the scheduler picks a process $i \in \set{0,\ldots, N-1}$.
Given a chosen process~$i$, the new state remains
$\mathbf{v}$ if the chosen process~$i$ did not hold a token ($\mathbf{v}(i) =~\bot$).
If $\mathbf{v}(i) = \top$, the new state is
$\mathbf{v}$ with probability~$\frac{1}{2}$ and 
$\mathbf{v} \xor e_i \xor e_{(i+1)\bmod N}$ with probability~$\frac{1}{2}$.
Here, $e_i$ denotes a~vector with $\top$ in the $i$-th position and $\bot$ everywhere else,
and $\xor$ is the XOR operation.
We want to ensure that, starting from an arbitrary initial assignment of tokens,
any population 
self-stabilizes 
with probability~1.

Process fairness for Herman's protocol is a set of $N$
atomic fairness requirements, each asserting that the process $i$ is executed
infinitely often, for each $i \in \{1,\ldots,N\}$.
Unfortunately, Herman's protocol does \emph{not} terminate with probability 1 
against some fair schedulers. 
To see this, consider the start state $s_0 = 
(\top,\bot,\top)$. Let us call the token held by Process 0 ``the first token'', and the
token held by Process 2 ``the second token''.
    Define a \defn{round} 
    as the following sequence of moves by the scheduler: keep choosing the 
    process that holds the first token until it passes the token to the
    right, and do the same to the second token. For example, the
    two configurations obtained after completing the first and second rounds
    from $s_0$ are, respectively, $(\top,\top,\bot)$ and $(\bot,\top,\top)$. To see that the
    scheduler is fair, for each integer $i 
    > 0$, the probability that the $i$-th round is not completed is 0 since
    the probability that one of the tokens will be kept at the same 
    process for an infinite amount of time is 0. Therefore, the probability that
    some round is not completed is also 0. Completing two rounds ensure that all
    the processes are picked. Therefore, every process will be chosen with 
    probability~1. On the other hand, observe that correct configurations are not
    seen in the induced DTMC, showing that self-stabilization holds with
    probability 0 under this scheduler.

Herman's protocol can be shown to self-stabilize with probability 1 under all 
fin-fair schedulers, which can be proved by our fully-automatic verification algorithm (presented later in the paper).
\OMIT{
Loosely speaking, for each bound $k \in \N$, 
each process will be chosen at least once every $kN$ moves of the scheduler.
Therefore, there is a non-zero probability that each token but the rightmost
gets passed to the right and that the rightmost token stays. 
}
\OMIT{
each token
has a non-zero probability to get passed to the right when chosen, and that the
rightmost token (i.e. held by process with the largest ID) has a non-zero
probability to get 
}

\newpage
\vspace{-3.0mm}
\section{Regular Model Checking: A Symbolic Framework} \label{sec:symbolic}
\vspace{-1.0mm}

\enlargethispage{4mm}

In this section, we recall \defn{regular model checking} (see e.g. 
\cite{Parosh12,rmc-thesis,anthony-thesis}), a symbolic framework for specifying 
infinite-state 
systems based on finite automata and regular transducers and developing automatic
verification (semi-)algorithms. 


A transition system $\struct = \transysMDP$ is specified in the 
framework 
as a regular language $\transysDom$ (e.g. as 
a~regular expression over some alphabet $\ialphabet$), and two 
``regular relations'' $\to_1, \to_2 \subseteq \ialphabet^* \times 
\ialphabet^*$. For simplicity, in the following we will assume that 
$\transysDom = \ialphabet^*$.
How do we specify regular relations? One standard way is to
restrict to length-preserving relations (i.e. the relation may only contain
a pair of words of the same length) and specify such relations as regular
languages over the alphabet $\ialphabet \times \ialphabet$. There is, then,
a simple one-to-one correspondence between the set of words over $\ialphabet 
\times \ialphabet$ and the set of all pairs of words over $\ialphabet$ of the 
same length. This can be achieved by mapping a pair $(v,w)$ of words
        $\ialphabet$ with $\Length{v} = \Length{w} = n$ to a word
        $v \otimes w$,
        defined as $(v_1,w_1)(v_2,w_2)\cdots
        (v_n,w_n)$ whenever $v = v_1\cdots v_n$ and $w  = w_1\cdots w_n$.


Proving that a property $\varphi$ holds over a transition system $\struct$ is 
done ``in a regular 
    way,'', by finding a ``regular proof'' for the property.
    For~example, if $\varphi$ asserts that the set $\mathit{Bad}$ of bad states can never
    be reached, then a regular proof amounts to finding an inductive invariant
    $\mathit{Inv}$ in the form of a regular language
\cite{Parosh12,rmc-thesis} that does not 
        intersect with $\mathit{Bad}$, i.e., $\mathit{Bad} \cap \mathit{Inv} = \emptyset$, $S_0 
    \subseteq \mathit{Inv}$ ($S_0$ is a regular set of initial states), and
    $post_{\to}(\mathit{Inv}) \subseteq \mathit{Inv}$, where ${\to} = {\to_1 \cup \to_2}$. Since
    regular languages are effectively closed under boolean operations and 
    taking pre/post images w.r.t.\ regular transducers, an algorithm for
    verifying the correctness of a given regular proof can be obtained by
    using language inclusion algorithms for regular 
    automata, e.g., \cite{BP13,abdulla-antichain}. The framework of regular proofs
    is incomplete in general since it could happen that there is a proof, but 
    no regular proof. 
    The pathological cases when only non-regular proofs 
    exist do not, however, seem to frequently occur in practice, e.g., see
\cite{BH06,BHRV12,Neider13,Parosh12,TL10,fast,BFLS05,rmc-thesis,Lin12-fsttcs}.

The framework of regular proofs has been extended to deal with almost-sure
termination for weakly-finite probabilistic concurrent programs in \cite{LR16}. 
We briefly summarise the main idea, since we reduce the fair termination
problem to their setting. 
By Proposition \ref{prop:removeProb}, 
the actual probability values do not matter 
in proving almost-sure termination.
For this reason, we may specify a weakly-finite MDP 
$\struct = \transysMDP$ as a regular specification in the same way as we
specify a non-probabilistic transition system in our regular specification
language. Given an MDP $\struct = \transysMDP$, a set $I_0 \subseteq
V_1$ of initial configurations, and a set $F \subseteq V_1$
of final configurations, a regular proof for
$\Prob(s_0 \models F) = 1$ for each $s_0 \in I_0$ is a pair
$\langle \mathit{Inv}, \prec\rangle$ consisting of a regular inductive invariant 
$\mathit{Inv} \subseteq \transysDom$ and a regular relation $\prec\ \subseteq
\transysDom \times \transysDom$ such that:
\begin{enumerate}
\item $I_0 \subseteq \mathit{Inv}$ and $\mathit{post}_{\to}(\mathit{Inv}) \subseteq \mathit{Inv}$.
\item $\prec$ is a strict preorder on $\transysDom$, i.e., it is
    irreflexive ($\forall s \in \transysDom: s \not\prec s$) and transitive
        ($\forall s,s',s'' \in \transysDom: s \prec s' \wedge
        s' \prec s'' \implies s \prec s''$).
\item irrespective of the nondeterministic transitions from any configuration
    in $\mathit{Inv}$, there is a probabilistic transition to a configuration
        in $\mathit{Inv}$ that decreases its rank with respect to $\prec$:
              \begin{align*}
              &\forall x \in \mathit{Inv}\setminus F, y \in S\setminus F:
              &~  \big(
              (x \to_1 y) ~\Rightarrow~
              (\exists z \in \mathit{Inv}:\;
              (y \to_2 z) \wedge x \succ z)
              \big)~.
              \end{align*}
\end{enumerate}
An automata-theoretic algorithm can then be devised for checking the
above verification conditions with respect to a given regular proof
\cite{LR16}.

\begin{example}{\textbf{[Herman's protocol, continued]}}
We provide a regular encoding of Herman's protocol.
    The configurations are words over the alphabet
    $\{\top,\bot,\overline{\top},\overline{\bot}\}$, 
where $\top$
(resp.~$\bot$) signifies that a process holds (resp.~does not hold) a token,
while overlining the character signifies that the process is chosen by the
scheduler. 
We set $\Sigma = \{\top, \bot\}$.
The set $S_0$ of initial configurations is $\ialphabet^* \top \ialphabet^*$, i.e.,
at least one process holds a token. The set of final configurations is
$\bot^* \top \bot^*$, i.e., there is only a~single token in the system.
The actions of the scheduler is to choose 
a process; this can be expressed as the regular expression 
    $I^* (\biword{\top}{\overline{\top}} + \biword{\bot}{\overline{\bot}}) I^*$, where
$I$ denotes the regular language $\biword{\top}{\top} +
\biword{\bot}{\bot}$. 
    \OMIT{
Without loss of generality, we assume that every process is always enabled
(i.e., can be chosen); if not, we simply add an idle transition.
The probabilistic transitions are idle in the case the chosen process does not
hold a token, or, if the chosen process holds a~token, it can either pass 
the token to the right (and merge tokens if the right-hand side neighbour
already holds one) or keep it.
}
The probabilistic actions can be expressed as a union of the following 
    three regular expressions:
    \begin{align}
        I^* (\biword{\overline \top}{\top} + \biword{\overline \bot}{\bot}) I^* &
        && \text{\bfseries (idle)} & \nonumber \\
        I^* \biword{\overline
        \top}{\bot}(\biword{\bot}{\top})+\biword{\top}{\top}) I^*, & 
        \quad 
        (\biword{\bot}{\top}+\biword{\top}{\top})I^*\biword{\overline \top}{\bot})  &&
        \text{\bfseries (pass token right)} \nonumber
    \end{align}
\end{example}

\OMIT{
This automatic procedure for verifying
    ``a regular proof'' gives rise to a fully-automatic verification. In
    the case of safety, such an algorithm works as follows: go 
        through all possible finite automata over $\ialphabet$ in some
        order and check whether 
    it is an inductive invariant that does not intersect with $\mathit{Bad}$. This
    sound procedure is also complete with respect to regular proofs, which
    often suffice in practice. This
    naive algorithm can be substantially improved by applying computational
    learning, e.g., see \cite{Neider13}.
}

\vspace{-3.0mm}
\section{Handling Fairness Requirements} \label{sec:transform}
\vspace{-1.0mm}

We now describe the main result of the paper: a general method for embedding
finitary fairness into regular model checking for probabilistic concurrent 
systems.

\vspace{-2.0mm}
\subsection{Regular Specifications of Fairness}
\vspace{-1.0mm}

When a complex system or a distributed protocol is being modelled in regular
model checking, it is often necessary to add an \emph{infinite} number of 
fairness
requirements. This is because such a system admits a finite but arbitrary 
number of 
agents or processes, each with its own fairness requirement
(e.g. that the process should be executed infinitely often). For this reason, 
it is not enough to
simply express the fairness requirements as a finite set of pairs of
regular languages (one for the premise, and one for the consequence). 
We describe a regular way of specifying infinitely many fairness 
constraints. Our presentation is a generalisation of the regular specification 
of fairness from~\cite{LTL-MSO,rmc-thesis}.

The general idea 
is
to define a ``regular function'' $\Tra$ that maps 
a configuration $s = s_1 \cdots s_n \in \transysDom$ to a word 
$w = w_1\cdots w_n$ such that $w_i$ contains:
(1) a bit $b_i$ indicating whether $s$ is in the premise of the $i$-th fairness
requirement, (2) a bit $b_i'$ indicating whether $s$ is in the consequence of 
the $i$-th fairness requirement, and (3) a bit $t$ indicating whether the $i$-th 
fairness requirement is justice or compassion. Such a regular specification
of fairness allows an infinite number of fairness constraints since 
$\transysDom$ is potentially infinite (i.e., containing words of unbounded 
lengths), though only the first $|s|$ fairness requirements matter for a word
$s \in \transysDom$. This is
sufficient for weakly-finite MDPs since the set of reachable configurations 
from any given configuration $s$ is finite and so, among the infinite number of 
fairness constraints, only finitely many are distinguishable. 
The regular
function can be defined by a letter-to-letter transducer with input alphabet 
$\ialphabet$ and output alphabet $\Gamma := \{0,1\} \times \{0,1\} \times
\{0,1\}$. Without loss of
generality, we assume that the $i$-th letter in the output of
every input word of $\Tra$ agree on the third bit (i.e., whether the
fairness requirement is justice or compassion is well-defined): for every
$s,s' \in \transysDom$ and $i \in \N$, if $\Tra(s)[i] = (a,b,c)$ and 
$\Tra(s')[i] = (a',b',c')$, then $c = c'$. Observe this condition on $\Tra$ 
can be algorithmically checked by using a simple automata-theoretic method:
find two accepted words in which in some position their third bits differ.

In this case, $\Tra$ gives rise
to compassion requirements $\SF$ and justice requirements~$\WF$ by associating
the $i$-th position in all output words by a unique fairness constraint. More
precisely,
let
\begin{itemize}
  \item  $A_i = \{ s : \Tra(s)[i] = (1,j,t), \text{ for some $j,t \in \{0,1\}$} \}$ and
  \item  $B_i = \{ s : \Tra(s)[i] = (j,1,t),  \text{ for some $j,t \in \{0,1\}$} \}$.
\end{itemize}
Define:
\begin{itemize}
  \item [(i)]
$\WF = \{ \Diamond\Box A_i \Implies \Box\Diamond B_i: 
        \Tra(s)[i] = (i,j,0), \text{ for some $s \in \transysDom$} 
        \text{ and $j \in \{0,1\}$}\}$,
        \item[(ii)]
$\SF = \{ \Box\Diamond A_i \Implies \Box\Diamond B_i: \Tra(s)[i] = 
(i,j,1), \text{ for some $s \in \transysDom$} \text{ and $j \in 
\{0,1\}$}\}$.
\end{itemize}
Therefore, by Proposition~\ref{prop:removeProb},
our regular fairness specification allows us to define 
weakly-finite fair MDPs $\FMDP$. In the following, we shall call such fair MDPs
\defn{regular}. 

Our main theorem is a regularity-preserving
reduction from proving almost sure termination for regular FMDPs (under
finitary fairness) to
proving almost sure termination for regular MDPs (without fairness).
\OMIT{
assuming
realizability of fairness.}
\begin{theorem}
    Let $\struct = \FMDP$ be a~regular representation of an~FMDP, $I_0
    \subseteq V_1$ be a~regular
    set of initial configurations, and
    $F \subseteq V_1$ be a regular set of final configurations.
    Then one can compute 
    a regular representation of MDP $\struct' = \langle \transysDom = V_1' \cup
    V_2'; \leadsto_1,
    \leadsto_2\rangle$ and two regular sets $I_0', F' \subseteq V_1'$
    such that it holds that
        if $\SF$ and $\WF$ are realizable, then
        $\Prob_{\struct'}( I_0' \models \Diamond F' ) = 1$ iff
            $\Prob_{\struct}(I_0 \models \Diamond F ) = 1$.
    \label{th:main}
\end{theorem}

\vspace{-2.0mm}
\subsection{Abstract Program Transformation}
\vspace{-1.0mm}
Before proving Theorem \ref{th:main}, let us first recall an abstract 
program transformation \emph{\`{a} la} Alur \& Henzinger \cite{AH98}, which 
encodes finitary fairness into a 
program using integer counter variables. Intuitively, we reserve one variable
for each atomic fairness condition
as an ``alarm clock'' that will set off if its corresponding process has not 
been executed for a long time, and one global variable $n$ that acts as a
\emph{default} value to reset a clock to 
as soon as the corresponding process is executed. Although Alur \& Henzinger 
\cite{AH98} did not discuss about probabilistic programs, their transformation
can be easily adapted to the setting of MDPs, though correctness still has to be
proven.

\OMIT{
To determine whether $\Prob_{\struct}(S_0 \models F) = 1$,
we will construct an MDP $\struct' = 
\langle \transysDom = V_1' \cup V_2'; \leadsto_1, \leadsto_2\rangle$ with a 
probability distribution $\delta'$ and two sets $S_0', F' \subseteq V_1'$
satisfying the following property: 
    \[
        \Prob_{\struct}(S_0 \models F) = \Prob_{\struct'}(S_0' \models F'),
    \] 
Hence, determining the value of $\Prob_{\struct'}(S_0' \models F')$ allows
us to determine the value of $\Prob_{\struct}(S_0 \models F)$.
}

We now elaborate on the details of the transformation. Given an FMDP 
$\struct = \FMDP$ with a probability distribution $\delta$, 
the transformation will produce an MDP $\struct' = \langle \transysDom = V_1' 
\cup V_2'; \leadsto_1, \leadsto_2\rangle$ with a probability distribution 
$\delta'$ as follows. 
Introduce a set $\Var$ of
``counter'' variables that range over natural numbers: $x_j$ (for each $j \in 
\WF$), $y_c$ (for each $c \in \SF$), and $n$. Let $\funcClass$ be
the set of all valuations $f$ mapping each variable in 
$\Var$ to a natural number
such that $f(x_j), f(y_c) \leq f(n)$ for each $j \in \WF$ and $c \in \SF$.
We define $V_1' = V_1 \times \funcClass$ and $V_2' = V_2 \times \funcClass$. 
\OMIT{
The transition relation
$\leadsto_2$ is simply a copy of $\to_2$ with the valuation component untouched,
i.e., $(s,f) \leadsto_2 (s',f')$ iff $s \to_2 s'$ and $f = f'$. 
We define 
the probability distribution 
    $\delta'((s,f),(s',f')) = \delta(s,s')$ if $f= f'$ and
    $\delta'((s,f),(s',f')) = 0$ otherwise.
}
We now define the transition relation $\leadsto_i$ such that $(s,f) \leadsto_i (s',f')$ if
$s \to_i s'$ and 
\begin{itemize}
    \item for each $z \in \Var$, $f(z) > 0$,
    \item $f'(n) := f(n)$,
    \item $x_j$ (for $j = (A,B) \in \WF$) and $y_c$ (for $c = (A,B) \in \SF$) change as follows:
    \begin{align*}
    f'(x_j) =& \left\{ \begin{array}{cl}
            f(x_j) - 1\!\!\!\! & \text{\quad if $s \in A \cap \overline{B}$} \\
                  f(n) & \text{\quad if $s \in \overline{A} \cup B$}
                              \end{array}
                    \right.
                    &&&
        f'(y_c) =& \left\{ \begin{array}{cl}
                  f(n) & \text{\quad if $s \in \overline{A} \cap \overline{B}$} \\
                  f(y_c) - 1\!\!\!\! & \text{\quad if $s \in A \cap \overline{B}$} \\
                  f(n) & \text{\quad if $s \in B$}
                              \end{array}
                    \right.
    \end{align*}
\end{itemize}
($\overline{A}$ denotes the set-complement of $A$).
Finally, 
we define the probability distribution $\delta'$ underlying 
$\leadsto_2$ as $\delta'((s,f), (s',f')) = \delta(s,s')$ whenever $s\in V_2$. 


Intuitively, the variables 
$x_j$'s and $y_c$'s keep track of how long the scheduler has delayed
choosing an enabled process, while the variable $n$ (unchanged once the initial
configuration of the MDP is fixed) aims to ensure that the scheduler is
$n$-fair. Since $n$ is a~variable (not a constant), the resulting MDP $\struct'$
captures precisely the behaviour of $\struct$ under fin-fair schedulers.

\begin{lemma}
    If $\struct$ is a weakly-finite FMDP, then $\struct'$ is weakly-finite. 
    \label{lm:weak-finite-preserve}
\end{lemma}

\begin{proof}
Since $\struct$ is weakly-finite, once a configuration $(s,f)$ of 
$\struct'$ is chosen, there are only finitely many different valuation 
$s' \in \transysDom$ such that $(s',f')$ (for some $f' \in \funcClass$)
is reachable from $(s,f)$.
In the following, we show that there are also only
finitely many different valuations $f' \in \funcClass$ such that $(s',f')$ 
(for some $s' \in \transysDom$) is reachable from $(s,f)$.
Let $X$ be the (finite) set $X = \mathit{post}_{\to^*}(s)$.
Define two
equivalence relations $\sim_{(s,f),c}$ and $\sim_{(s,f),j}$ on the set
$2^{\transysDom} \times 2^{\transysDom}$ of pairs of subsets of $\transysDom$
as follows: 
\begin{itemize}
    \item $(A,B) \sim_{(s,f),c} (A',B')$ iff (a) $A \cap X = A' \cap X$ and
        $B \cap X = B' \cap X$, and (b)~$(A,B) \in \SF$ iff $(A',B') \in
        \SF$.
    \item $(A,B) \sim_{(s,f),j} (A',B')$ iff (a) $A \cap X = A' \cap X$ and
        $B \cap X = B' \cap X$, and (b)~$(A,B) \in \WF$ iff $(A',B') \in \WF$.
\end{itemize}
Observe that, since $X$ is finite, both equivalence relations are of finite 
index (i.e.~have only finitely many equivalence classes).
This implies that we need not distinguish two variables in $\Var\setminus\{n\}$
if they are both for the justice or both for the compassion requirements, in the same equivalence class in the
appropriate relation $\sim_{(s,f),c}$ or $\sim_{(s,f),j}$, and they both have
the \emph{same} initial $f$-values. To see this, let $\leadsto {} :=  {}
\leadsto_1 \cup \leadsto_2$.
Observe that,
for each $(s',f') \in post_{\leadsto^*}((s,f))$ and $c = (A,B), c' = (A',B') \in
\SF$, it is the case that $f'(y_c) = f'(y_{c'})$ iff $f(y_c) = f(y_{c'})$.
Similarly,
for each $(s',f') \in post_{\leadsto^*}((s,f))$ and $j = (A,B), j' = (A',B') 
\in \WF$, it is the case that $f'(x_j) = f'(x_{j'})$ iff $f(x_j) = f(x_{j'})$.
In other words, identical counter values across similar fairness constraints
remain identical under an application of $\leadsto$. Since all counter values
in all reachable configurations $(s',f') \in post_{\leadsto^*}((s,f))$
are in $\{0,\ldots,f(n)\}$, it immediately follows that
$post_{\leadsto^*}((s,f))$ is finite.
\end{proof}

We next state a correctness lemma for the transformation.
To this end,
given a set $S_0 \subseteq \transysDom$ of initial configurations in $\struct$, 
we define:
\begin{itemize}
    \item $S_0' := S_0 \times \funcClass_=$, where $\funcClass_=$ contains
        functions $f \in \funcClass$ such that $f(x_j) = f(y_c) = f(n)$
        for each $j \in \WF$ and $c \in \SF$.
    \item $F' = (F \times \funcClass_{> 0}) \cup (\transysDom \times 
            \funcClass_0)$,
            where $\funcClass_0$ contains all $f \in \funcClass$ such that
            $f(x_j) = 0$ for some $j \in \WF$ or $f(y_c) = 0$ for some $c
            \in \SF$ (i.e. one of the alarms has been triggered), and
            $\funcClass_{>0} := \funcClass \setminus \funcClass_0$.
\end{itemize}
\begin{lemma}[Correctness]
    If $\struct$ is a~weakly-finite FMDP, 
    it is the case that
        $$\Prob_{\struct}(S_0 \models \Diamond F) = 
        \Prob_{\struct'}(S_0' \models \Diamond F').$$
    \label{lm:fair-for-wfmdp}
\end{lemma}

\begin{proof}
In this proof, we make use of the following notation.
For a sequence $\pi$ of pairs $(x_1,y_1), (x_2,y_2),\ldots$,
we use the notation $\proj{1}(\pi)$ (resp. $\proj{2}(\pi)$) to denote the 
sequence $\pi$ projected to the first (resp. second) arguments, i.e.,
$x_1,x_2,\ldots$ (resp. $y_1,y_2,\ldots$).
Moreover,
for each $k > 0$, we define $\funcClass_k$ to be the set of all
functions $f \in \funcClass_=$ such that $f(n) = k$.

We first prove that $\Prob_{\struct}(S_0 \models \Diamond \finals) \geq
\Prob_{\struct'}(S_0' \models \Diamond \finals')$, i.e., the transformation does
not increase the probability of reaching final states.
For each $k \in \N$, 
consider a $k$-fair scheduler $\sigma$ for $\struct$.
It suffices to prove that given any $s_0 \in S_0$,
$\Prob_{\struct_\sigma}(s_0 \models \Diamond \finals) = 
\Prob_{\struct_{\sigma'}'}((s_0,f) \models \Diamond \finals')$, where $f \in 
\funcClass_k$ and $\sigma'(\pi) := \sigma(\proj{1}(\pi))$ (note that
$\funcClass_k$ contains exactly one $f$ compatible with $s_0$).
In turn, to prove this, it suffices to prove that the DTMC $\struct_\sigma$
restricted to configurations reachable from $s_0$ is isomorphic to
$\struct_{\sigma'}'$ restricted to configurations reachable from $(s_0,f)$.
This can be seen from the
fact that configurations of the form $(S\setminus \finals) \times
\funcClass_0$ are not reachable from $(s_0,f)$; if they were reachable, since
the counter encoding precisely emulates the definition of finitary fairness
\cite{AH98},
the witnessing path $\pi$ would give rise to a path $\proj{1}(\pi)$ that would witness
that $\sigma$ is not $k$-fair, contradicting our original assumption.

We next prove that $\Prob_{\struct}(S_0 \models \Diamond \finals) \leq
\Prob_{\struct'}(S_0' \models \Diamond \finals')$, i.e., our transformation does
not decrease the probability of reaching final states.
Consider any $(s_0,f) \in 
S_0'$ and any scheduler $\sigma'$ on $\struct'$. 
Consider the scheduler $\sigma$ on $\struct$ that simulates the behaviour of 
$\sigma'$, but as soon as one of the alarm clocks has set off the scheduler 
goes through all consequence sets (say, $X_1,\ldots, X_m$ for some $m \in \N$;
the sequence is finite since $\struct$ and $\struct'$ are weakly finite) in some order
and chooses actions that satisfy them in a round robin manner (which can be done
since we consider process constraints). More precisely, for each path $\pi$
in $\struct_{\sigma'}'$, define $\sigma(\proj{1}(\pi)) := \sigma'(\pi)$. 
For
each path $\pi$ ending in a configuration in $(S \setminus \finals) \times
\funcClass_0$, the action of the scheduler on any path with $\pi$ as a prefix is
to loop through $X_1,\ldots,X_m$ and pick actions that satisfy them.
Therefore, the scheduler $\sigma$ is $K$-fair for $K := 2m$.
Furthermore,
consider the two tree-shaped DTMCs $\struct_1$ and $\struct_2$, where
$\struct_1$ is obtained from $\struct_\sigma$
by restricting the sets of configurations to those that are reachable from
$s_0$, and $\struct_2$ is obtained from $\struct_{\sigma'}$ by restricting the
sets of configurations to those that are reachable from $(s_0,f)$.
$\struct_1$ and $\struct_2$ are isomorphic except for subtrees
$\run_\pi$ where $\pi$ is a path in~$\struct_2$ from $(s,f)$
ending in $(S\setminus\finals) \times \funcClass_0$ without visiting a
configuration in $\finals \times \funcClass_{>0}$.
The probability $p$ of visiting a configuration in $\finals'$ in $\struct_2$ from
$(s,f')$ on the condition that $\pi$ is taken is 1.
Thus, 
on the condition that the prefix $\proj{1}(\pi)$ is taken, the probability of 
visiting a configuration in $\finals'$ in $\struct_1$ cannot exceed $p$.
This shows us that $\Prob_{\struct_\sigma}(s_0 \models \Diamond \finals)
\leq \Prob_{\struct_{\sigma'}'}((s_0,f) \models \Diamond \finals')$.
Consequently, since the choice of $(s_0,f) \in S_0'$ and $\sigma'$ was
arbitrary, we can conclude that $\Prob_{\struct}(S_0 \models \Diamond \finals) \leq
\Prob_{\struct'}(S_0' \models \Diamond \finals')$. 
\end{proof}

\noindent
These two lemmas immediately imply Theorem \ref{th:main}.
\vspace{-2.0mm}
\subsection{Finitary Fairness in Regular Model Checking}
\vspace{-1.0mm}

We now show how to implement the aforementioned abstract program
transformation in our regular model checking framework. 
Fix a regular representation of an FMDP $\struct = \FMDP$, which includes
two automata over the alphabet $\ialphabet \times \ialphabet$ representing
$\to_1$ and $\to_2$, and an
automaton over the alphabet $\ialphabet \times \Gamma$ representing
the regular specification of the fairness conditions $\SF$ and $\WF$.
[Recall that $\Gamma := \{0,1\} \times \{0,1\} \times \{0,1\}$.]
We describe the construction of $\leadsto_1$ (the 
construction for $\leadsto_2$ is similar). Let
$\Aut = (\ialphabet \times \ialphabet, \controls,\transrel,q_0,\finals)$
be an automaton representing $\to_1$ and
$\Aut^f = (\ialphabet \times \Gamma, \controls^f,\transrel^f,q_0^f,\finals^f)$
be an automaton representing the regular specification of fairness. The construction
of the automaton for $\leadsto_1$ has two stages.

\enlargethispage{4mm}

\vspace{-2mm}
\paragraph{Stage 1: compute an intermediate automaton.} The intermediate
automaton $\AutB$ will have the alphabet $\ialphabet' := 
(\ialphabet \times \ialphabet) \cup \Gamma$ and recognize a subset of
$((\ialphabet \times \ialphabet)\Gamma)^*$. 
Intuitively, on input $(a,b)
\in \ialphabet \times \ialphabet$, the automaton $\AutB$ simultaneously
takes a transition over $(a, b)$ in $\Aut$ and any transition $(a,c)$ in $\Aut^f$,
proceeding into an intermediate state where it remembers the value of $c$, which it outputs in the next step.
This process is repeated until both $\Aut$ and $\Aut^f$ accept.
More precisely, the automaton is defined as 
    $\AutB := (\ialphabet',\controls^B,\transrel^B,q_0^B,\finals^B)$
where:
\begin{itemize}
    \item $\controls^B = \controls \times \controls^f \times 
            (\Gamma \cup \{?\})$, 
    $q_0^B =  (q_0,q_0^f,?)$, and
    $\finals^B = \finals \times \finals^f \times \{?\}$
    \item $\transrel_B$ has the following transitions:
        \begin{itemize}
            \item $((p_1,q_1^f,?),(a,b),(p_2,q_2^f,c))$ if
                $(p_1,(a,b),p_2) \in \transrel$ and $(q_1^f,(a,c),q_2^f) \in 
                \transrel^f$.
            \item $((p,q^f,c),c,(p,q^f,?))$ for each $c \in \Gamma$.
        \end{itemize}
\end{itemize}

\newcommand{\onecnt}[0]{\bullet}
\newcommand{\zerocnt}[0]{\circ}

\paragraph{Stage 2: regular substitution of letters in $\Gamma$.}
Our encoding of counters is unary using symbols $\onecnt$ and $\zerocnt$, where
$\onecnt$ represents a pebble, and $\zerocnt$ represents empty space.
For instance, a number $n \in \mathbb{N}$ is encoded as $\onecnt^n \zerocnt^*$ (the
number of $\zerocnt$'s is arbitrary, though the length of the whole encoding is
constant due to our use of length-preserving transducers).
We define the following regular languages for manipulating the counters:
\begin{itemize}
    \item (Identity) $\ID := (\onecnt,\onecnt)^+(\zerocnt,\zerocnt)^*$,
    \item (Decrement) $\DEC := (\onecnt,\onecnt)^*(\onecnt,\zerocnt)(\zerocnt,\zerocnt)^*$, and
    \item (Reset) $\RES := (\onecnt,\onecnt)^+(\zerocnt,\onecnt)^*$.
\end{itemize}
Define the regular substitution $\sigma$ mapping letters in $\Gamma$ to
regular languages:
\begin{itemize}
    \item if $(x,y,z)$ is $(i,1,j)$ or $(0,i,0)$ (for
        $i,j \in \{0,1\}$), then $\sigma((x,y,z)) = \RES$.
    \item if $(x,y,z)$ is of the form $(1,0,i)$ (for some $i \in \{0,1\}$),
        then $\sigma((x,y,z)) = \DEC$.
    \item define $\sigma((0,0,1)) = \ID$.
\end{itemize}
We then apply the regular substitution $\sigma$ to the letters $\Gamma$ 
appearing in our intermediate automaton $\AutB$. The resulting automaton
implements the desired automaton for $\leadsto_1$.

\vspace{-1mm}
\paragraph{Finishing off the rest of the construction.}
Computing $S_0'$ and $\finals'$ is easy.
Define $S_0'$ to be the set of all words $a_1 w_1 a_2 w_2\cdots a_m w_m$ such
that $a_1\cdots a_m \in S_0$ and $w_i \in \onecnt^+$ for each $i \in
\{1,\ldots,m\}$.
Similarly, define $\finals'$ to be the set of all words $a_1 w_1 a_2 w_2 \cdots
a_m w_m$ such that
\begin{itemize}
  \item  either $a_1 \cdots a_m \in \finals$ and $w_i \in (\onecnt^+
    \zerocnt^*) \cup \zerocnt^+$ for each $i \in \{1,\ldots,m\}$, or
  \item  $w_i \in \zerocnt^+$ for some $i \in \{1,\ldots,m\}$.
\end{itemize}
Finite automata for these sets could be easily constructed given automata for
$S_0$ and $\finals$.


\vspace{-1mm}
\begin{example}{\textbf{[Herman's protocol]}}
We encode process fairness in the following way.
The counters use the unary encoding, their values represented as the lengths
    of sequences of $\onecnt$'s padded on the right by the symbol $\zerocnt$ (crucial to 
    keep the transducers length-preserving). For example, the number 3 is
    represented by any word of the form $\onecnt \onecnt \onecnt~\zerocnt^*$.
    Define $\CTR = \onecnt^* \zerocnt^*$, i.e., the set of all valid counters.
The set of initial configurations can be expressed using the regular expression
$(\ialphabet \cdot \CTR)^* (\top \cdot \CTR) (\ialphabet \cdot \CTR)^*$, i.e., counters for all processes are initialized to an arbitrary
value.
The set of final configurations is now
$(\bot \cdot \CTR)^* (\top \cdot \CTR) (\bot \cdot \CTR)^* \cup
(\Sigma \cdot \CTR)^* (\Sigma \cdot\ \zerocnt^*) (\Sigma \cdot \CTR)^*$, i.e.,
either there is exactly one token in the system, or (at least) one counter has reached~0.
Scheduler now also performs operations on the counters for processes: for
a~chosen process, the counter is reset, for other processes, the counter is decremented.
This can be expressed as the language
$(I \cdot \DEC)^* \big(((\bot,\overline \bot) + (\top, \overline \top))\cdot
\RES\big)(I \cdot \DEC)^*$.
Actions of the protocol are the same as in the original encoding and the values of counters are left unmodified:
\begin{align*}
(I \cdot \ID)^* \big(((\overline \bot, \bot) + (\overline \top, \top))\cdot \ID\big)(I \cdot \ID)^* &&& \text{\bfseries (idle)} \\
(I \cdot \ID)^* \big((\overline \top, \bot)\cdot
    \ID\big)\big(((\bot, \top) + (\top,\top))\cdot \ID\big)(I \cdot \ID)^* &&& \text{\bfseries (pass token right$_1$)} \\
\big(((\bot, \top) + (\top,\top))\cdot \ID\big)(I \cdot \ID)^* \big((\overline \top, \bot)\cdot
    \ID\big) &&& \text{\bfseries (pass token right$_2$)}
\end{align*}
%
At this point, we can use existing tools for checking termination (without fairness
constraints), e.g.~\cite{LR16}.
Indeed, we can automatically check that the system after reduction terminates
with probability one, thus proving that Herman's protocol fairly terminates
with probability one (under finitary process-fair schedulers).
\end{example}

\vspace{-3.0mm}
\section{Implementation and Experiments} \label{sec:experiments}
\vspace{-1.0mm}


The approach presented in this paper has been implemented in the tool \toolname.%
\footnote{%
\url{https://github.com/uuverifiers/autosat/tree/master/Fairness}
}
For evaluation, we extracted models of a~number of probabilistic parameterized
systems.
The tool receives a~system with fairness conditions and transforms it into a
system without fairness conditions, where fairness of the original system is
encoded using counters.
For solving liveness in the output transformed system, we use
\slrp{}~\cite{LR16} (in the \emph{incremental liveness proofs} setting)
as the underlying liveness checker for parameterized systems.

\begin{table}
\begin{center}
\vspace{-5.5mm}
\label{tab:results}
\caption{Times of analyses of probabilistic paremeterised systems.
The timeout was set to 10~hours (timeout is denoted as T/O).}
\footnotesize
\hspace{-3mm}
\begin{minipage}{5.4cm}
  \begin{tabular}{|l||r|}
  \hline
  Case study & Time \\
  \hline\hline
  Herman's protocol (merge, line)   & 3.64\,s \\
  Herman's protocol (annih., line)  & 4.33\,s \\
  Herman's protocol (merge, ring)   & 4.31\,s \\
  Herman's protocol (annih., ring)  & 4.61\,s \\
  \hline
  Moran process (2 types, line)     & 2\,m 48\,s \\
  Moran process (3 types, line)     & 56\,m 14\,s \\
  \hline
  Cell cycle switch (1 types, line) & 43.94\,s \\
  Cell cycle switch (2 types, line) & 9\,h 46\,m \\
  \hline
  Clustering (2 types, line)        & 10\,m 30\,s \\
  Clustering (3 types, line)        & T/O \\
  \hline
  Coin game ($k = 3$, clique)       & 1\,m 0\,s \\
  \hline
  \end{tabular}
\end{minipage}
\end{center}
\end{table}

Table~\ref{tab:results} shows the results of our experiments.
The times given are the wall clock times for the individual
benchmarks on a~PC with 4 Quad-Core AMD Opteron 8389
processors with Java heap memory limited to 64\,GiB.
The time included translation of the system into a~system without fairness (always
less than 1s)
and the runtime of \slrp{}.

\paragraph{Herman's protocol.}
We consider two versions of \emph{Herman's protocol} 
(described in more detail in Section~\ref{sec:herman})
that differ in the way how they
handle the situation when a~process already holding a~token receives another
token---either the two tokens are merged into one, or they are both
annihilated---and two topologies: a~line and a~ring.
In our version of Herman's protocol, all processes are always enabled, and in
case they do not hold a~token, they can only stay in their state\footnote{%
Note that even for the line topology, this model requires fairness to verify
that a~configuration with a~single token is reachable with probability~1.
The model used in~\cite{LR16} did not require the fairness assumption since it
modelled processes without tokens as disabled.}.
Fig.~\ref{fig:herman-results} shows an example of a~synthesized solution for
the model of Herman's protocol on a~ring (the \emph{annihilating} variant).

\vspace{-0.0mm}
\paragraph{Moran process.}\label{sec:label}
\vspace{-0.0mm}
Our second example uses variants of \emph{Moran process},
a model of genetic drift, on a~linear array~\cite{Moran58}, 
where the order in which individuals move is determined by the scheduler.
Each node of the array is an allele, and there are two types of 
alleles: $A$ and $B$ (we can generalise this to any number $N \geq 2$ of 
alleles). 
At each round, the scheduler picks an allele. 
The allele will randomly either copy its type to itself
or ``infect'' one of the neighbours (copy its type to a random 
neighbour). 
We check the fair termination property that the system eventually reaches
a~``drift'' state, where all alleles are of the same type, 
under every process-fair scheduler, with probability~one. 

%
Note that the termination property is not true if we do not impose fairness:
for example, when in
$A A B B$, an unfair scheduler can simply choose the first $A$ all the time.
Additionally, small variants of the model may not satisfy the property.
Consider the variation where the chosen allele must infect one of its neighbors.
For a linear array of size at least~4, a~fair
scheduler can play in such a way that guarantees that
the system never terminates in a~drift state.

\vspace{-0.0mm}
\paragraph{Cell cycle switch.}\label{sec:label}
\vspace{-0.0mm}

The model of \emph{cell cycle switch} is a~simplification of the model
of~\cite{cell-cycle-switch}
to reach a~common decision of all members of a~population between two possible
outcomes that approximately matches the initial relative majority.
We assume cells are of three types according to their decision: $X$, $Y$, and
\emph{undecided} (in the case of two types, we consider $X$ and
\emph{undecided}).
A~decided cell ($X$ or $Y$) can change the type of an~undecided
neighbour to its own.
In addition, a~decided cell can change a~neighbour with an opposite decision to
\emph{undecided}.
We verify that from any initial configuration,
with probability~1, the system stabilizes into a~configuration where all cells
share a~common decision.

\begin{figure}[t]
\begin{minipage}[b]{6cm}
\includegraphics[keepaspectratio,width=\textwidth]{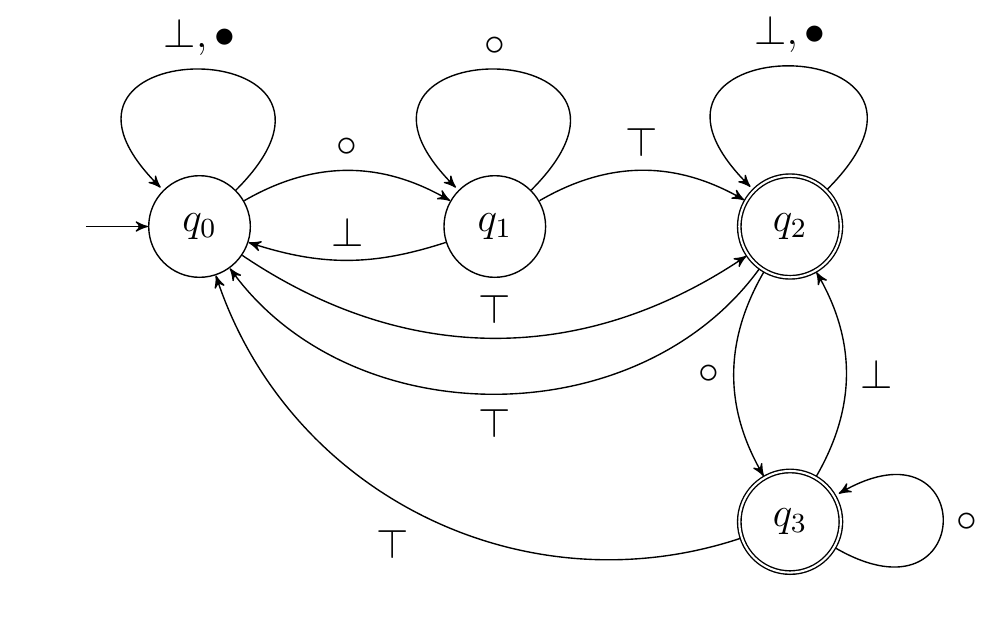}

  \begin{center}
  \vspace{-6mm}
  a) $\mathit{Inv}$
  \end{center}
\end{minipage}
\begin{minipage}[b]{6cm}
  \includegraphics[keepaspectratio,width=\textwidth]{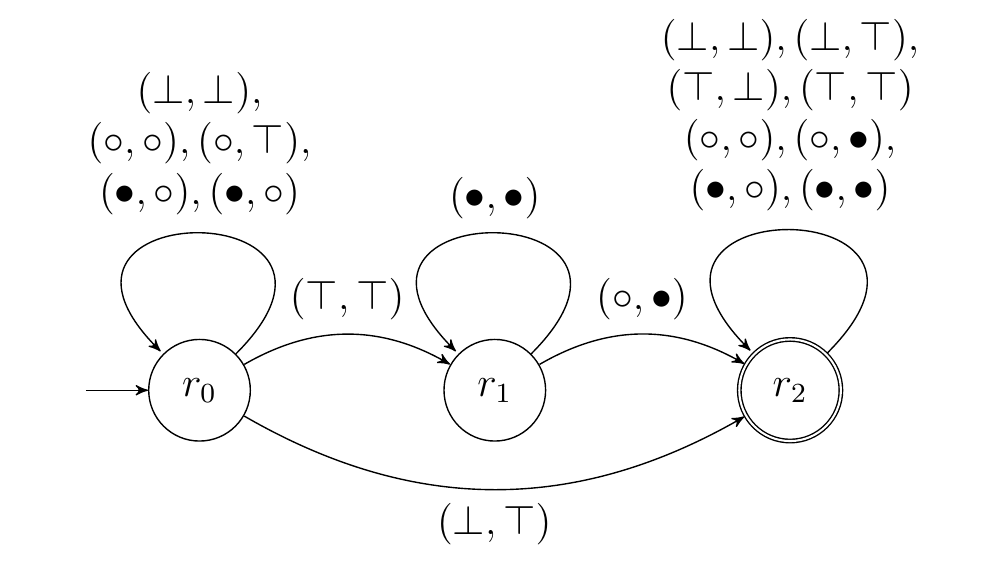}

  \begin{center}
  b) $\prec$
  \end{center}
\end{minipage}
\caption{Synthesized advice bits for Herman's protocol on a~ring}
\label{fig:herman-results}
\end{figure}

\vspace{-0.0mm}
\paragraph{Clustering.}\label{sec:label}
\vspace{-0.0mm}
The \emph{clustering} example considers a~population model of alleles of 2
(resp.~3) types, say $\{A,B\}$ (resp.\ $\{A, B, C\}$), on a~line.
The alleles can change position with their neighbours of a~different type,
e.g.~$AB \to BA$.
We verify that from any configuration, the system will reach
a~state where the alleles form 2 (resp.\ 3) clusters of the same type.

\vspace{-0.0mm}
\paragraph{Coin game.}\label{sec:label}
\vspace{-0.0mm}
In the \emph{coin game} use case, we consider a~population protocol where every
agent has one of two types of coins: \emph{gold} or \emph{silver}.
In each step, an agent chosen by the scheduler will either keep its currency,
or switch to the currency held by the majority from $k$~randomly selected
neighbours.
We verify that we eventually get to a~configuration where all agents have the
same type of coins.

\bigskip

The experiments show that our encoding of fairness into systems is viable and
can be used for verification of parameterized systems with fairness by their
reduction to systems without fairness.
On the other hand, when the size of the regular proof is large, we observe that
the problem for the underlying solver gets significantly more difficult (as can
be seen on the example of \emph{clustering} for three types of alleles).
We conjecture that the performance can be improved significantly by making the
solver take into account the (not arbitrary) structure of the problem, which we
leave for future work.



\smallskip\noindent\emph{Future work.} We leave the reader with several
research challenges. A natural question is how to deal with non-finitary 
fairness for parameterized probabilistic concurrent systems in general and 
in the framework of regular model checking. Secondly, since there are
numerous examples of population models over more complex topologies (e.g.
grids), how do you deal with termination and fair termination over such models
in the parameterized setting? 

\smallskip\noindent\emph{Acknowledgement.} We thank anonymous reviewers
and Dave Parker for their helpful feedback. This work was supported by the Czech
Science Foundation (project 16-24707Y), the BUT FIT project
FIT-S-17-4014, the IT4IXS: IT4Innovations Excellence in Science project
(LQ1602), Yale-NUS Starting Grant, the European Research Council under 
ERC Grant Agreement no. 610150, and Swedish Research Council (2014-5484).


\bibliographystyle{splncs}

\bibliography{references}

\clearpage

%

\end{document}